\documentclass{elsarticle}

\makeatletter
\def\ps@pprintTitle{%
   \let\@oddhead\@empty
   \let\@evenhead\@empty
   \def\@oddfoot{\reset@font\hfil\thepage\hfil}
   \let\@evenfoot\@oddfoot
}
\makeatother

\usepackage{lineno,hyperref,url}
\usepackage{amsfonts,amssymb,amsmath,setspace,graphicx}
\usepackage{amsthm}
\usepackage{color}

\DeclareMathOperator*{\argmax}{argmax}
\newtheorem{theorem}{Theorem}

\newtheorem{lemma}{Lemma}
\newtheorem{definition}{Definition}
\renewcommand{\d}{\mathrm{d}}
\newcommand{\p}{p}
\newcommand{\ignore}[1]{}

\newcommand{\F}{\mathcal{F}}
\renewcommand{\H}{\mathcal{H}}
\newcommand{\C}{\mathcal{C}}
\newcommand{\Px}{\mathcal{P}}
\renewcommand{\P}{\mathbb{P}}
\newcommand{\E}{\mathbb{E}}
\newcommand{\reals}{\mathbb{R}}
\newcommand{\nats}{\mathbb{N}}
\newcommand{\NS}{\mathbb{NS}}
\newcommand{\NSrm}{\mathrm{NS}}
\newcommand{\ind}{\mathbb{I}}
\newcommand{\eps}{\epsilon}
\renewcommand{\dim}{d}
\renewcommand{\k}{k}
\newcommand{\X}{\mathcal{X}}
\newcommand{\Y}{\mathcal{Y}}
\newcommand{\Z}{\mathcal{Z}}
\newcommand{\target}{f}

\newcommand{\accept}{{\textsc{Accept}}}
\newcommand{\reject}{{\textsc{Reject}}}
\newcommand{\dist}{\rho}
\newcommand{\distance}{\dist}

\begin{document}
\begin{frontmatter}

\title{Testing Piecewise Functions}

\author{Steve Hanneke}
\ead{steve.hanneke@gmail.com}

\author{Liu Yang}
\ead{liu.yang0900@outlook.com}

\begin{abstract}
This work explores the query complexity of property testing for 
general piecewise functions on the real line, in the active and passive 
property testing settings.  The results are proven under an abstract 
zero-measure crossings condition, which has as special cases 
piecewise constant functions and piecewise polynomial functions.
We find that, in the active testing setting, the query complexity 
of testing general piecewise functions is independent of the 
number of pieces.  We also identify the optimal dependence 
on the number of pieces in the query complexity of passive testing
in the special case of piecewise constant functions.
\end{abstract}


\end{frontmatter}



\section{Introduction}

Property testing is a well-studied class of problems, in which an algorithm must 
decide (with high success probability) from limited observations of some object whether a given property
is satisfied, or whether the object is in fact far from all objects with the property \cite{rubinfeld:96,goldreich:98}.
In many property testing problems, it is natural to describe the object as a 
function $\target$ mapping some instance space $\X$ to a value space $\Y$, and the property
as a family of functions $\F$.  The property testing problem is then equivalently stated
as determining whether $\target \in \F$ or whether $\target$ is far from all elements of $\F$,
for a natural notion of distance.  In this setting, the observations are simply $(x,\target(x))$ points.
The setting has been studied in several variations, depending on how the $x$ observation points are 
selected (e.g., at random or by the algorithm). 
In particular, this type of property testing problem is closely
related to the PAC \emph{learning} model of \cite{valiant:84}.\footnotemark{\hskip 1mm}  
One of the main theoretical questions in the study of property testing is how many observations
are needed by the optimal tester, a quantity known as the \emph{query complexity} of the testing problem.  
In the above context, this question is most interesting when one
can show that the query complexity of the testing problem is significantly smaller than the 
query complexity of the corresponding PAC learning problem \cite{vapnik:74,blumer:89,ehrenfeucht:89,hanneke:16a,hanneke:fntml,dasgupta:05}.

\footnotetext{In the 
PAC (\emph{probably approximately correct}) learning model proposed by 
Valiant \cite{valiant:84} (and previously by Vapnik and Chervonenkis \cite{vapnik:74}),
it is assumed that $\target \in \F$, and an algorithm is tasked with choosing any $\hat{f}$ that is (with high probability) 
within distance $\eps$ of $\target$, given access to a finite number of $(x,\target(x))$ pairs drawn at random 
(or in some variants, selected by the algorithm).}

The property testing literature is by now quite broad, and includes testing algorithms and characterizations of 
their query complexities for many function classes $\F$; see \cite{ron:08,ron:10,goldreich:17} for introductions to this literature and some of its key techniques.
However, nearly all of the work on the above type of property testing problem has focused on the special 
case of \emph{binary} functions, where $|\Y|=2$, or in some cases with $\Y$ a general finite field.  
In this article, we are interested in the study of more general types of functions, including real-valued functions.
Previous works on testing real-valued functions include testers for whether $\target$ is 
monotone \cite{goldreich:00,ergun:00,berman:14},
unate \cite{baleshzar:17}, 
or Lipschitz \cite{berman:14,jha:13}.
In the present work, we study a general problem of testing \emph{piecewise} functions.
Specifically, we consider a scenario where $\X = \reals$, and where the function class 
$\F$ can be described as a $\k$-\emph{piecewise} function (for a given $\k \in \nats$), 
where each piece is a function in a given base class of functions $\H$.
Formally, defining $t_{0} = -\infty$ and $t_{\k} = \infty$, for any $t_{1},\ldots,t_{\k-1} \in \reals$ with $t_{1} \leq \cdots \leq t_{\k-1}$, 
and any $h_{1},\ldots,h_{\k} \in \H$, define (for any $x \in \reals$)
\begin{equation*}
f\!\left( x ; \{h_{i}\}_{i=1}^{\k}, \{t_{i}\}_{i=1}^{\k-1} \right) = h_{i}(x) \text{ for the $i$ such that } t_{i-1} < x \leq t_{i}.
\end{equation*}
Then we consider classes $\F = \F_{\k}(\H)$ defined as 
\begin{equation*}
\F_{\k}(\H) = \left\{ f\!\left(\cdot; \{h_{i}\}_{i=1}^{\k}, \{t_{i}\}_{i=1}^{\k-1} \right) : h_{1},\ldots,h_{\k} \in \H,  t_{1} \leq \cdots \leq t_{\k-1} \right\}.
\end{equation*}
In the results below, we will be particularly interested in the dependence on $\k$ in the 
query complexity of the testing problem.  To be clear, the values $t_{1},\ldots,t_{\k-1}$ and functions $h_{1},\ldots,h_{\k}$ 
are all free parameters, with different choices of these yielding different functions, all contained in $\F_{\k}(\H)$.
Thus, in the testing problem (defined formally below), 
the algorithm may directly depend on $\k$ and $\H$, but in the case of $\target \in \F_{\k}(\H)$ the specific 
values of $t_{1},\ldots,t_{\k-1}$ and functions $h_{1},\ldots,h_{\k} \in \H$ specifying $\target$ are all considered as \emph{unknown}.

In this work, our primary running example is the scenario where $\Y = \reals$ and $\H$ is the set of \emph{degree-$\p$ polynomials}: 
$\H = \{ x \mapsto \sum_{i=0}^{\p} \alpha_{i} x^{i} : \alpha_{0},\ldots,\alpha_{\p} \in \reals \}$,
so that $\F_{\k}(\H)$ is the set of $\k$-piecewise degree-$\p$ polynomials.
A further interesting special case of this is when $\p=0$, in which case $\F_{\k}(\H)$ is the set of $\k$-piecewise \emph{constant} functions.
However, our general analysis is more abstract, and will also apply to many other interesting function classes $\H$. 

Specifically, for the remainder of this article, 
we consider $\Y$ as an arbitrary nonempty set (equipped with an appropriate $\sigma$-algebra), 
and $\H$ as an arbitrary nonempty set of measurable functions $\X \to \Y$
satisfying the following\break \emph{zero-measure crossings} property:
\begin{equation}
\label{eqn:zero-measure-crossings}
\forall h,h^{\prime} \in \H,~~ h \neq h^{\prime} \implies \lambda( \{ x : h(x) = h^{\prime}(x) \} ) = 0,
\end{equation}
where $\lambda$ denotes the Lebesgue measure.
It is well known that this property is satisfied by polynomial functions: 
for any two polynomial functions $h,h^{\prime}$ of degree $\p$, 
if they agree on a nonzero measure set of points, then we can find in that set distinct points $x_{1},\ldots,x_{\p+1}$ 
on which they agree, but since the values on any $\p+1$ distinct points uniquely determine the polynomial function,
it must be that $h = h^{\prime}$.
Thus, the analysis below indeed applies to piecewise polynomial functions as a special case.
The zero-measure crossings property is also satisfied by many other interesting function classes, 
such as shifted sine functions $\{ x \mapsto \sin(x + t) : t \in \reals \}$ 
or normal pdfs $\{ x \mapsto c \cdot e^{-(x-t)^{2}/2} : t \in \reals \}$.




The testing problem is defined as follows.
Fix a value $\eps \in (0,1)$ and a probability measure $\Px$ over $\X$, 
and for any measurable $f,g : \X \to \Y$, define $\dist(f,g) = \Px(x : f(x) \neq g(x))$, 
the $L_{0}(\Px)$ (pseudo)distance between $f$ and $g$.
Further define $\dist(f,\F) = \inf_{g \in \F} \dist(f,g)$, the distance of $f$ from a set of functions $\F$.
In the active testing protocol, the algorithm samples a number $s$ of iid unlabeled examples from $\Px$,
and then interactively queries for the $\target$ values at points of its choosing from among these: that is, 
it chooses one of the $s$ examples $x$, queries for its value $\target(x)$, then selects another of the $s$ examples $x^{\prime}$,
queries for its value $\target(x^{\prime})$, and so on, until it eventually halts and produces a decision of either \accept~or \reject.
The following definition is taken from \cite{yang:12}.

\begin{definition}
\label{defn:tester}
An $s$-sample $q$-query $\eps$-tester for $\F$ under the distribution $\Px$ is a randomized algorithm $A$
that draws a sample $S$ of size at most $s$ iid from $\Px$, sequentially queries for the value of $\target$ on at most $q$ points of $S$,
and satisfies the following properties (regardless of the identity of $\target$):
\begin{itemize}
\item if $\target \in \F$, it decides \accept~with probability at least $\frac{2}{3}$.
\item if $\dist(\target,\F) \geq \eps$, it decides \reject~with probability at least $\frac{2}{3}$.
\end{itemize}
\end{definition}

We will be interested in identifying values of $q$ for which there exist $s$-sample $q$-query $\eps$-testers for $\F$, 
known as the \emph{query complexity}; we are particularly interested in this when $s$ is polynomial in $\k$, $1/\eps$, and the complexity of $\H$ (defined below).
In the special case that $q=s$, so that the algorithm queries the values at \emph{all} of the points in $S$,
the algorithm is called a \emph{passive tester} \cite{goldreich:98}, whereas in the general case of $q \leq s$ it is referred to as an \emph{active tester} \cite{yang:12} 
(in analogy to the setting of \emph{active learning} studied in the machine learning literature \cite{hanneke:fntml,settles:12}).\\

\noindent {\bf Remark on general distributions $\mathbf{\Px}$:}
For simplicity, we will focus on $\Px = {\rm Uniform}(0,1)$ in this work.
However, all of our results easily generalize to all distributions $\Px$ over $\reals$
absolutely continuous with respect to Lebesgue measure.
Specifically, \cite{yang:12} discuss a simple technique 
which produces this generalization, by using the empirical distribution to effectively rescale the 
real line so that the distribution appears approximately uniform.  One can show that this technique is also applicable 
in the present more-general context as well.  The rescaling effectively changes the function class $\H$, 
but its graph dimension (defined below) remains unchanged, and the zero-measure crossings property is preserved due to 
$\Px$ being absolutely continuous.  In the special case of testing piecewise \emph{constant} functions, 
even the restriction to absolutely continuous distributions can be removed, as then the zero-measure crossings
property always holds.  The interested reader is referred to \cite{yang:12} for the details of this rescaling technique.

\subsection{Related Work: Testing Unions of Intervals}
\label{subsec:uint}


The work of \cite{yang:12} explored the query complexity of testing in both the active and passive models,
for a variety of function classes, but all under the restriction $\Y = \{0,1\}$.  Of the results of \cite{yang:12}, the
most relevant to the present work are results on the query complexity of 
testing \emph{unions of intervals}: $x \mapsto \ind\left[ x \in \bigcup_{i = 1}^{n} [t_{2i-1},t_{2i}] \right]$, for a fixed $n \in \nats$, 
defined for all nondecreasing sequences $t_{1},\ldots,t_{2n} \in \reals \cup \{\pm \infty\}$.
They specifically find that the query complexity of testing unions of intervals 
is $O(1/\eps^{4})$ --- independent of $n$ --- in the active testing setting, 
and is $O(\sqrt{n}/\eps^{5})$ in the passive testing setting, with an $\Omega(\sqrt{n})$ lower bound; 
these results strengthened an earlier result of Kearns and Ron \cite{kearns:00}.
%
Note that unions of intervals are a special case of piecewise constant functions,
and indeed the techniques we employ in constructing the tester below 
closely parallel the analysis of unions of intervals by \cite{yang:12}.  
However, to extend that approach to general piecewise functions --- even piecewise \emph{constant} functions --- requires
careful consideration about what the appropriate generalization of the technique should be.
In particular, the original proof involved a \emph{self-correction} step, wherein $\target$ is 
replaced by a smoothed function, which is then rounded again to a binary-valued function.
This kind of smoothing and rounding interpretation no longer makes sense for general $\Y$-valued
functions.  We are therefore required to re-interpret these steps in the more general setting,
where we find that they can be reformulated as \emph{voting} on the function value at each point  
(rather than rounding a smoothed version of $\target$).  In the end, the active tester below 
for general $\F_{\k}(\H)$ functions has significant differences compared to the original tester 
for unions of intervals from \cite{yang:12}.  Nevertheless, we do recover the dependences 
on $\k$ established by \cite{yang:12} for active and passive testing of unions of intervals, 
as special cases of our results on testing piecewise constant functions.  Our results 
on testing general piecewise functions further extend this beyond piecewise constant 
functions, and require the introduction of additional uniform concentration arguments from VC theory.
These considerations about general piecewise functions also lead us to an appropriate generalization 
of the notion of \emph{noise sensitivity}, a quantity 
commonly appearing in the property testing literature for binary-valued functions \cite{yang:12}.

It is also worth mentioning that, for binary-valued functions in higher dimensions $\reals^{n}$, 
the noise sensitivity has also been used to test for the property that the decision boundary 
of $\target$ has low \emph{surface area} \cite{kothari:14}.  For simplicity, the present work focuses on 
the one-dimensional setting, leaving for future work the appropriate generalization of these results 
to higher-dimensional spaces.

\subsection{The Graph Dimension}

For any set $\Z$, any collection $\C$ of subsets of $\Z$, and any $m \in \nats \cup \{0\}$,
following \cite{vapnik:71}, we say $\C$ \emph{shatters} a set $\{z_{1},\ldots,z_{m}\} \subseteq \Z$
if 
\[
| \{ C \cap \{z_{1},\ldots,z_{m}\} : C \in \C \} | = 2^{m}.
\]
The VC dimension of $\C$ is then 
defined as the largest integer $m$ for which $\exists \{z_{1},\ldots,z_{m}\} \subseteq \Z$ shattered by $\C$,
or as infinity if no such largest $m$ exists.

The VC dimension is an important quantity in characterizing the optimal sample complexity of statistical learning
for binary classification.  It is also useful in our present context, for the purpose of defining a related complexity
measure for general functions.  Specifically, the \emph{graph dimension} of $\H$, denoted $\dim$ below, 
is defined as the VC dimension of the collection $\{ \{ (x,h(x)) : x \in \X \} : h \in \H \}$ (where $\Z = \X\times\Y$ in this context).
For the remainder of this article, we restrict to the case $0 < \dim < \infty$,
but otherwise 
we can consider $\H$ as a completely \emph{arbitrary} set of functions (subject to \eqref{eqn:zero-measure-crossings}).

To continue the examples from above, we note that in the case of 
$|\Y| \geq 2$ and $\H$ as the set of constant functions (so that $\F_{\k}(\H)$ is the $\k$-piecewise constant functions) 
one can easily show $\dim = 1$.
Moreover, in the case of $\Y = \reals$ and $\H$ as the set of degree-$\p$ real-valued polynomial functions
(so that $\F_{\k}(\H)$ is the $\k$-piecewise degree-$\p$ polynomial functions),
it follows from basic algebra that $\dim = \p+1$, since any polynomial $f$ is uniquely
determined by any $\p+1$ distinct $(x,f(x))$ pairs (so that $\dim \leq \p+1$), 
and any $\p+1$ pairs $(x,y)$ (with distinct $x$ components) can be fit by a polynomial 
(so that $\dim \geq \p+1$, by choosing $\p+1$ distinct $x$ points, 
each with two corresponding $y$ values, 
and all $2^{\p+1}$ choices of which $y$ value to use for each point can be fit by a degree-$\p$ polynomial).

The results below will be expressed in terms of $\k$, $\eps$, and $\dim$.
The dependence of the query complexity on each of these is an important
topic to consider.  However, in the present work, we primarily focus on identifying
the optimal dependence on $\k$.  
The optimal \emph{joint} dependence on $\k$, $\eps$, and the complexity of $\H$ is a problem we leave for future work.
%

\subsection{Main Results}

We are now ready to state our main results.  
We present their proofs in the sections below.
The first result is for active testing.  Its proof in Section~\ref{sec:active} below is based 
on an analysis of a novel generalization of the notion of 
the \emph{noise sensitivity} of a function.

\begin{theorem}
\label{thm:active}
For any $\eps \in (0,1/2)$, there exists an 
$s$-sample $q$-query $\eps$-tester for $\F_{\k}(\H)$ under the distribution ${\rm Uniform}(0,1)$,
with 
$s = O\!\left( \frac{\dim \k}{\eps^{6}} \ln\!\left(\frac{1}{\eps}\right) \right)$ 
and 
$q = O\!\left( \frac{\dim}{\eps^{8}}\ln\!\left(\frac{1}{\eps}\right) \right)$.
\end{theorem}

In particular, this immediately implies that the optimal dependence on $\k$ in the query complexity of active testing is $O(1)$.
This independence from $\k$ in the query complexity is the main significance of this result.\footnote{Since the dependence on $\eps$
in this result is greater in the bound on $q$ than in $s$ (due to some over-counting in the bound on $q$ when $\eps$ is small), 
we should note that this result also implies 
the existence of an $s$-sample $\min\{q,s\}$-query $\eps$-tester, since we could simply query all $s$ samples and then 
simulate the interaction with the oracle internally.  This becomes relevant for $\eps \ll 1/\sqrt{\k}$.  That said, we describe 
a simple testing strategy at the end of Section~\ref{sec:active} which obtains 
$s = q = O\!\left( \frac{\dim \k \ln(\k)}{\eps}\ln\!\left(\frac{1}{\eps}\right) \right)$, 
which is superior in this range $\eps \ll 1/\sqrt{\k}$ anyway.}

Our second result is for \emph{both} active and passive testing, and applies specifically to the special case of piecewise \emph{constant} functions
(for any $\Y$ space).  The upper bound in this result is again based on a generalization 
of the notion of the \emph{noise sensitivity} of a function, and recovers as a special case 
the result of \cite{yang:12} for unions of intervals.  The lower bound is essentially already known, 
as it was previously established by \cite{kearns:00,yang:12} for unions of intervals, which are a special 
case of piecewise constant functions.  The proof of the lower bound 
for general piecewise constant functions follows immediately 
from this via a simple reduction argument.
The complete proof of this theorem is presented in Section~\ref{sec:constant} below.

\begin{theorem}
\label{thm:constant}
If $\H$ is the set of \emph{constant} functions, then 
for any $\eps \in (0,1/2)$, there exists an  
$s$-sample $q$-query $\eps$-tester for $\F_{\k}(\H)$ under the distribution\break 
${\rm Uniform}(0,1)$,
with $s = O\!\left( \frac{\sqrt{k}}{\eps^{5}} \right)$, and 
with $q = O\!\left( \frac{1}{\eps^{4}} \right)$ for active testing, and $q = s$ for passive testing.

Moreover, 
in this case, if 
$\eps \in (0,1/8)$, 
\emph{every} $s$-sample $s$-query $\eps$-tester for $\F_{\k}(\H)$ under the distribution ${\rm Uniform}(0,1)$ 
has $s = \Omega\!\left(\sqrt{\k}\right)$.
\end{theorem}

In particular, this implies that the optimal query complexity of passive testing of $\k$-piecewise constant functions 
has dependence $\sqrt{\k}$ on the number of pieces $\k$.
It is also straightforward to extend this lower bound to $\k$-piecewise degree-$\p$ polynomial functions. 
However, our results below do not imply an upper bound with $\sqrt{\k}$ dependence on $\k$ for passive testing of this larger function class,
and as such identifying the optimal query complexity of passive testing for piecewise degree-$\p$ polynomials 
(and for general classes $\H$ satisfying \eqref{eqn:zero-measure-crossings}) 
remains an interesting open problem.



\section{A Generalization of Noise Sensitivity}

Here we develop a generalization of the definition of the
\emph{noise sensitivity} used by \cite{yang:12} in their analysis of testing unions of intervals;
this will be the key quantity in the proofs of the above theorems.
Throughout this section, we let $\Px$ be the Lebesgue measure restricted to $[0,1]$: i.e., the distribution ${\rm Uniform}(0,1)$.
For any $x \in \X$ and $y \in \Y$, define $\H_{(x,y)} = \{h \in \H : h(x)=y\}$.
Let $x \sim {\rm Uniform}(0,1)$, and conditioned on $x$, 
let $x^{\prime} \sim {\rm Uniform}(x-\delta,x+\delta)$.
Define the \emph{instantaneous noise sensitivity} at $x$ as\footnote{The original definition 
of \cite{yang:12} essentially defined $\NSrm_{\delta}(f,x) = \P( f(x^{\prime}) \neq f(x) | x )$.  The 
involvement of $\H$ in our generalization of the definition will be crucial to the results below.} 
\begin{equation*}
\NSrm_{\delta}(f,x;\H) 
= \inf_{h \in \H_{(x,f(x))}} \P( h(x^{\prime}) \neq f(x^{\prime}) | x),
\end{equation*}
or in the event that $\H_{(x,f(x))}$ is empty, define $\NSrm_{\delta}(f,x;\H) = 1$.
Then define the \emph{noise sensitivity} as 
\begin{equation*}
\NS_{\delta}(f;\H) = \E\!\left[ \NSrm_{\delta}(f,x;\H) \right] = \int_{0}^{1} \NSrm_{\delta}(f,z;\H) {\rm d}z.
\end{equation*}

The instantaneous noise sensitivity essentially measures the ability of functions from $\H$ 
to match the behavior of $f$ in a local neighborhood around a given point $x$, and the  
noise sensitivity is simply the average of this over $x$.

We have the following two key lemmas on this definition of noise sensitivity.
Their statements and proofs directly parallel the analysis of unions of intervals 
by \cite{yang:12}, but with a few important changes (particularly in the proof of Lemma~\ref{lem:NSnonInt-piecewise}) 
to generalize the arguments to general piecewise functions.

\begin{lemma}
\label{lem:NSInt-piecewise}
For any $\delta > 0$, 
$\forall f \in \F_{\k}(\H)$, 
 $\NS_{\delta}(f;\H) \le (\k-1) \frac{\delta}{2}$.
\end{lemma}
\begin{proof}
For $f \in \F_{\k}(\H)$, let $h_{1},\ldots,h_{\k} \in \H$ and $t_{1},\ldots,t_{\k-1} \in \reals$ be such that 
$f(\cdot) = f\!\left(\cdot;\{h_{i}\}_{i=1}^{\k},\{t_{i}\}_{i=1}^{\k-1}\right)$.
Thus, for any $i \leq \k$ and $x \in (t_{i-1},t_{i}]$, and any $y \in \reals$,
we have $h_{i}(x) = f(x)$, and we would have $h_{i}(y) \neq f(y)$ only if $x$ and $y$ are separated by one of the boundaries $t_{i-1}$ or $t_{i}$;
in particular, $x \leq t_{i} \leq y$ or $y \leq t_{i-1} \leq x$.

Let $x \sim {\rm Uniform}(0,1)$ and (conditioned on $x$) $y \sim {\rm Uniform}(x-\delta, x+\delta)$.
Denoting by $i(x)$ the $i$ with $t_{i-1} < x \leq t_{i}$, we have 
\[
\NSrm_{\delta}(f,x;\H) \leq \P( h_{i(x)}(y) \neq f(y) | x ) \leq \P( x \leq t_{i(x)} \leq y | x ) + \P( y \leq t_{i(x)-1} \leq x | x ),
\]
so that 
\begin{align*}
\NS_{\delta}(f;\H) 
& \leq \P( x \leq t_{i(x)} \leq y ) + \P( y \leq t_{i(x)-1} \leq x ) 
\\ & \leq \sum_{i=1}^{\k-1} \left( \P( x \leq t_{i} \leq y ) + \P( y \leq t_{i} \leq x ) \right),
\end{align*}
where the last inequality uses the facts that $\P( t_{\k} \leq y ) = 0$ and $\P( y \leq t_{0} ) = 0$.

For any fixed $t \in \reals$, 
\begin{equation*}
\P( x \leq t \leq y ) \leq \int_{0}^{\delta}  \P_{y^{\prime} \sim {\rm Uniform}(t-z-\delta,t-z+\delta)}[y^{\prime} \geq t] \d z = \int_0^\delta \frac{\delta-z}{2\delta} \d z = \frac{\delta}{4},
\end{equation*}
noting that, if $t$ is outside $[\delta,1]$, then the probability can only become smaller.
Similarly, any $t \in \reals$ has $\P(y \leq t \leq x) \leq \frac{\delta}{4}$, again noting that 
the probability only becomes smaller if $t$ is outside $[0,1-\delta]$.

Combining these inequalities with the above bound on $\NS_{\delta}(f;\H)$ yields $\NS_{\delta}(f;\H) \leq (\k-1) \frac{\delta}{2}$, as claimed.
%
\end{proof}

\begin{lemma}
\label{lem:NSnonInt-piecewise}
Fix any $\eps \in (0,1/2)$ and let $\delta = \frac{\eps^2}{32\k}$.  Let $f : \X \to \Y$ be any function with 
$\NS_{\delta}(f;\H) \leq (\k-1)\frac{\delta}{2}(1+\frac{\eps}{4})$.  Then $\dist(f,\F_{\k}(\H)) < \eps$.
\end{lemma}
\begin{proof}
Let $\k^{\prime} = \lfloor 1+(\k-1)(1+\frac{\eps}{2}) \rfloor$.
We first argue that $f$ is $\frac{\eps}{2}$-close to a function in $\F_{\k^{\prime}}(\H)$,
and then we argue that every function in $\F_{\k^{\prime}}(\H)$ 
is $\frac{\eps}{2}$-close to $\F_{\k}(\H)$.

For each $h \in \H$, consider the function $f_{\delta}^{h} : [0,1] \to [0,1]$ defined by
\[
f_{\delta}^{h}(x) 
= \frac{1}{2\delta} \int_{x-\delta}^{x+\delta} \ind[f(t)=h(t)] \d t.
\]

The function $f_{\delta}^{h}$ is the convolution of 
$t \mapsto \ind[f(t)=h(t)]$ 
and the uniform kernel
$\phi : \reals \to [0,1]$ defined by 
$\phi(x) = \frac{1}{2\delta} \ind[ |x| \le \delta ]$.  
%
Note that, since any distinct $h,h' \in \H$ have $\int_{0}^{1} \ind[h(t)=h'(t)] \d t = 0$ (by the zero-measure crossings assumption \eqref{eqn:zero-measure-crossings}),
the sum (over $h \in \H$) of all $f_{\delta}^{h}(x)$ values is at most $1$.  In particular, at most one
$h \in \H$ has $f_{\delta}^{h}(x) > 1/2$ for any $x$.

Fix $\tau = \frac{4}{\eps} \NS_{\delta}(f;\H)$.  
Since $\NS_{\delta}(f;\H) \leq (\k-1)\frac{\delta}{2}(1+\frac{\eps}{4}) < \frac{\eps^2}{32}$,
we have $\tau < 1/8$.
For each $x$, let $h_{x} = \argmax_{h \in \H} f_{\delta}^{h}(x)$ (breaking ties arbitrarily); 
since the sum (over $h \in \H$) of $f_{\delta}^{h}(x)$ values is finite (bounded by $1$), it follows that the value $\sup_{h \in \H} f_{\delta}^{h}(x)$
is actually realized by some $f_{\delta}^{h}(x)$ with $h \in \H$, so that $h_{x}$ is well-defined.  
Define a function $g^{*} : [0,1] \to \Y\cup\{*\}$ by 
$g^{*}(x) = h_{x}(x)$ if $f_{\delta}^{h_{x}}(x) \geq 1-\tau$,
and $g^{*}(x) = *$ otherwise.  
Next, define a function $g : \reals \to \Y$ by setting, for any $x \in [0,1]$, 
$g(x) = h_{z}(x)$ where $z$ is the largest value in $[0,x]$ for which $g^{*}(z) \neq *$,
and let $g_{x} = h_{z}$; if no such $z$ exists, take $z$ minimal in $[x,1]$ with $g^{*}(z) \neq *$ instead;
if that also does not exist, we can define $g(x)$ and $g_{x}$ arbitrarily, as this case will not come up in our present context.
As we discuss below, $f_{\delta}^{h_{x}}(x) = \sup_{h \in \H} f_{\delta}^{h}(x)$ is continuous in $x$, 
which entails that at least one of these two possible $z$ values will exist if $g^{*}$ is not everywhere equal $*$ in $[0,1]$.
For completeness, also define, for any $x < 0$, $g(x) = g_{0}(x)$, and for any $x > 1$, $g(x) = g_{1}(x)$.
\medskip

Now note that 
\begin{align}
\dist(f,g) & = \Px( x : f(x) \neq g(x) ) 
\leq \Px( x : g^*(x) = * ) + \Px( x : * \neq g^*(x) \neq f(x) ) 
\notag \\ & = \Px\!\left( x : \sup_{h \in \H} f_{\delta}^{h}(x) < 1 - \tau \right) 
\notag \\ & ~~+ \Px\!\left( x : \H \setminus \H_{(x,f(x))} \neq \emptyset, \sup_{h \in \H \setminus \H_{(x,f(x))} } f_{\delta}^{h}(x) \geq 1 - \tau \right). \label{eqn:dist-f-g-bound}
\end{align}
Because $\tau < 1/2$, at most one $h$ can have $f_{\delta}^{h}(x) \geq 1-\tau$ (as discussed above), 
so that
if the event $\sup_{h \in \H} f_{\delta}^{h}(x) < 1-\tau$ holds or the event $\H \setminus \H_{(x,f(x))} \neq \emptyset$ and $\sup_{h \in \H \setminus \H_{(x,f(x))}} f_{\delta}^{h}(x) \geq 1-\tau$ holds,
then either way we would have 
$\sup_{h \in \H_{(x,f(x))}} f_{\delta}^{h}(s) < 1-\tau$ (or $\H_{(x,f(x))} = \emptyset$); 
thus, since the two events implying this are disjoint, the sum of probabilities in \eqref{eqn:dist-f-g-bound} is at most 
\[
\Px\!\left( x : \H_{(x,f(x))} = \emptyset \text{ or } \sup_{h \in \H_{(x,f(x))}} f_{\delta}^{h}(x) < 1-\tau \right).
\]
Now observe that $\NSrm_{\delta}(f,x;\H) = 1 - \sup_{h \in \H_{(x,f(x))}} f_{\delta}^{h}(x)$ if $\H_{(x,f(x))} \neq \emptyset$,
and $\NSrm_{\delta}(f,x;\H) = 1$ if $\H_{(x,f(x))} = \emptyset$.
Together with Markov's inequality, this implies that  
\begin{multline*}
\Px\!\left( x : \H_{(x,f(x))} = \emptyset \text{ or } \sup_{h \in \H_{(x,f(x))}} f_{\delta}^{h}(x) < 1-\tau \right)
\\ = \Px( x : \NSrm_{\delta}(f,x;\H) > \tau )
 < \frac{\NS_{\delta}(f;\H)}{\tau} = \frac{\eps}{4}.
\end{multline*}
Thus, we have established that $\dist(f,g) \leq \frac{\eps}{4}$.  
\medskip

Next we show that $g \in \F_{m+1}(\H)$ for some nonnegative integer $m \le$\break $(\k-1) (1 + \frac{\eps}{2})$.
Since each $f_{\delta}^{h}$ is the convolution of $\ind[f(\cdot)=h(\cdot)]$ with a uniform kernel of width $2\delta$,
it is $\frac{1}{2\delta}$-Lipschitz smooth.  Also recall that 
$\tau < 1/2$, and the sum of all $f_{\delta}^{h}(x)$ values for a given $x$ is at most $1$.
Thus, if we consider any two points $x,z \in [0,1]$ with $g^{*}(x) \neq *$, $g^{*}(z) \neq *$,
$x < z$, and $h_{x} \neq h_{z}$, 
then it must be that $|x-z| \geq 2\delta (1-2\tau)$, and that 
there is at least one point $t \in (x,z)$ with $\sup_{h \in \H} f_{\delta}^{h}(t) = 1/2$.
Since each $f_{\delta}^{h}$ is $\frac{1}{2\delta}$-Lipschitz, so is $\sup_{h \in \H} f_{\delta}^{h}$,
so that we have 
\[
\int_{t - 2\delta (\frac{1}{2}-\tau)}^{t + 2\delta (\frac{1}{2}-\tau)} \sup_{h \in \H} f_{\delta}^{h}(s) \d s
\leq 2 \int_{0}^{2\delta (\frac{1}{2}-\tau)} \left(\frac{1}{2} + \frac{s}{2\delta}\right) \d s
= 2 \delta \left(\frac{1}{2}-\tau\right) \left( \frac{3}{2} - \tau\right).
\]
Therefore, 
\begin{align*}
& \int_{x}^{z} \NSrm_{\delta}(f,s;\H) \d s
\geq \int_{x}^{z} \left(1 - \sup_{h \in \H} f_{\delta}^{h}(s) \right) \d s
\\ & \geq (z-x) - 2 \delta \left( \frac{1}{2}-\tau \right) \left(\frac{3}{2} - \tau \right)
\geq 2 \delta \left( 1 - 2\tau\right) - 2 \delta\left( \frac{1}{2}-\tau\right) \left(\frac{3}{2} - \tau\right)
\\ & = 2 \delta \left(\frac{1}{2}-\tau \right) \left( \frac{1}{2} + \tau \right)
= 2 \delta \left( \frac{1}{4} - \tau^{2} \right).
\end{align*}
Since any $x$ with $g^{*}(x) \neq *$ has $g(x) = g^{*}(x)$,
and since $g_{t}$ is extrapolated from the left in $*$ regions of $g^{*}$ 
(aside from the case of an interval of $*$ values including $0$, where it is extrapolated from the right),
for every point $x > 0$ for which there exist arbitrarily close points $y$
having $g_{y} \neq g_{x}$, we must have that $g^{*}(x) \neq *$, 
and that there is a point $z < x$ such that $g^{*}(z) \neq *$ 
and such that every $t \in (z,x)$ has $g_{t} = g_{z} \neq g_{x}$.
Combined with the above, we have that $\int_{z}^{x} \NSrm_{\delta}(f,s;\H) \d s \geq 2 \delta( \frac{1}{4} - \tau^{2} )$.
Altogether, if $g$ has $m$ such ``transition'' points, then 
\[
\NS_{\delta}(f;\H) 
= \int_{0}^{1} \NSrm_{\delta}(f,s;\H) \d s
\geq m 2 \delta \left( \frac{1}{4} - \tau^{2} \right).
\]
By assumption, $\NS_{\delta}(f;\H) \leq (\k-1) \frac{\delta}{2} (1 + \frac{\eps}{4})$.
Therefore, we must have
\[
m \leq \frac{(\k-1) \delta (1 + \frac{\eps}{4})}{4 \delta (\frac{1}{4} - \tau^{2} )} 
\leq (\k-1) \frac{1+\frac{\eps}{4}}{1 - 4\tau^{2}} 
\leq (\k-1) \frac{1+\frac{\eps}{4}}{(1-2\tau)^{2}} 
\leq (\k-1) \left(1+\frac{\eps}{2}\right),
\]
since $\tau < 1/8$.
In particular, this means $g \in \F_{m+1}(\H)$ for an $m \leq (\k-1)(1+\frac{\eps}{2})$, as claimed.
\medskip


As a second step in the proof, we show that for any nonnegative integer $m \leq (\k-1) (1+\frac{\eps}{2})$, any function $g^{\prime} \in \F_{m+1}(\H)$
is $\frac{\eps}{2}$-close to a function in $\F_{\k}(\H)$.
Let $t_{1},\ldots,t_{m} \in \reals$ with $t_{1} \leq \cdots \leq t_{m}$, and $h_{1},\ldots,h_{m+1} \in \H$, be such that $g^{\prime}(\cdot) = f(\cdot ; \{h_{i}\}_{i=1}^{m+1},\{t_{i}\}_{i=1}^{m} )$.
For each $i \in \{1,\ldots,m+1\}$, let $\ell_{i} = \Px( (t_{i-1},t_{i}] )$ denote the probability mass in the $i^{{\rm th}}$ region.
In particular, $\ell_{1} + \cdots + \ell_{m+1} = 1$, so there must be a set $S \subseteq \{1,\ldots,m+1\}$ with $|S| = (m+1) - \k \leq (\k-1)\frac{\eps}{2}$ 
such that 
\[
\sum_{i \in S} \ell_{i} \leq \frac{(m+1) - \k}{(m+1)} \leq \frac{(\k-1)\eps/2}{1+(\k-1)(1 + \eps/2)} < \frac{\eps}{2}.
\]
Define a function $f^{\prime} : \X \to \Y$ such that, 
for each $i \in \{1,\ldots,m+1\}$ and $x \in (t_{i-1},t_{i}]$,
we set $f^{\prime}(x) = h_{j}(x)$ 
for the $j \in \{1,\ldots,m+1\} \setminus S$ of smallest $|i-j|$ (breaking ties to favor smaller $j$). 
The function $f^{\prime}$ is then contained in $\F_{\k}(\H)$, and has $f^{\prime}(x) = g^{\prime}(x)$ for every $x \in (t_{i-1},t_{i}]$ with $i \notin S$, 
and hence $\distance(g^{\prime},f^{\prime}) < \frac{\eps}{2}$.
This completes the proof, since taking $g^{\prime} = g$ yields 
$\distance(f,f^{\prime}) \leq \distance(f,g) + \distance(g,f^{\prime}) < \eps$.  
\end{proof}

\section{Active Testing}
\label{sec:active}

We can use the above lemmas to construct an active tester for $\F_{\k}(\H)$ as follows.
Fix any $\eps \in (0,1/2)$.
Let $m = \left\lceil \frac{c}{\eps^{4}} \right\rceil$, 
$\ell = \left\lceil \frac{c'\dim}{\eps^{4}}\ln\!\left(\frac{c''}{\eps}\right) \right\rceil$, 
$\delta = \frac{\eps^{2}}{32 \k}$, and 
$s = m + \left\lceil \max\!\left\{ \frac{2\ell}{\delta}, \frac{8}{\delta} \ln(12 m) \right\} \right\rceil$, 
for appropriate choices of numerical constants $c,c',c'' \geq 1$ from the analysis below.
Sample $s$ points $x_{1}^{\prime\prime},\ldots,x_{s}^{\prime\prime}$ independent ${\rm Uniform}(0,1)$.
Define $x_{i} = x_{i}^{\prime\prime}$ for each $i \leq m$.
For each $i \leq m$, denoting $t_{i0} = m$, for each $j \in \{1,\ldots,\ell\}$, 
let $t_{ij} = \min\{ t \in \{t_{i(j-1)}+1,\ldots,s\} : x_{t}^{\prime\prime} \in (x_{i} - \delta, x_{i}+\delta) \}$ if such a value exists 
(if it does not exist, the tester may return any response, as this is a failure case),
and define $x_{ij}^{\prime} = x_{t_{ij}}^{\prime\prime}$.
Thus, the random variables $x_{1},\ldots,x_{m}$ are iid ${\rm Uniform}(0,1)$ and, given $x_{i}$, 
the random variables $x_{i1}^{\prime},\ldots,x_{i\ell}^{\prime}$ are conditionally iid 
${\rm Uniform}( (x_{i}-\delta, x_{i}+\delta) \cap [0,1])$ (given $x_{i}$ and the event that they exist).
The tester requests the $f$ values for all 
$m (\ell+1)$ of these points $x_{i}$, $x_{ij}^{\prime}$, $i \in \{1,\ldots,m\}$, $j \in \{1,\ldots,\ell\}$.  
It then calculates, for each $i \leq m$,
\[
\widehat{\NSrm}_{\delta}(f,x_{i};\H) = \min_{h \in \H_{(x_{i},f(x_{i}))}} \frac{1}{\ell} \sum_{j=1}^{\ell} \ind[ h(x_{ij}^{\prime}) \neq f(x_{ij}^{\prime}) ],
\]
or $\widehat{\NSrm}_{\delta}(f,x_{i};\H) = 1$ in the event that $\H_{(x_{i},f(x_{i}))}$ is empty.
Then define 
\[
\widehat{\NS}_{\delta}(f;\H) = \frac{1}{m} \sum_{i=1}^{m} \widehat{\NSrm}_{\delta}(f,x_{i};\H).
\]

\begin{lemma}
\label{lem:NS-empirical}
If $\k \geq 80/\eps$, 
then for appropriate choices of numerical constants $c,c',c''$,
for any measurable function $f : \X \to \Y$, 
with probability at least $2/3$, all of the above $x_{ij}^{\prime}$ points exist, 
and the following two claims hold:
\begin{align*}
\NS_{\delta}(f;\H) \leq (\k-1) \frac{\delta}{2} & \implies \widehat{\NS}_{\delta}(f;\H) \leq (\k-1) \frac{\delta}{2} \left( 1 + \frac{\eps}{8} \right)
\\ \NS_{\delta}(f;\H) > (\k-1) \frac{\delta}{2} \left( 1 + \frac{\eps}{4} \right) & \implies \widehat{\NS}_{\delta}(f;\H) > (\k-1) \frac{\delta}{2} \left( 1 + \frac{\eps}{8} \right).
\end{align*}
\end{lemma}
\begin{proof}
For each $i \leq m$, any $t \in \{m+1,\ldots,s\}$ has conditional probability (given $x_{i}$) 
at least $\delta$ of having $x_{t}^{\prime\prime} \in (x_{i}-\delta,x_{i}+\delta)$.
Therefore, by a Chernoff bound (applied under the conditional distribution given $x_{i}$) and the law of total probability, 
with probability at least $1 - \exp\{ - \delta (s-m) / 8 \}$, the number of $t \in \{m+1,\ldots,s\}$ 
with $x_{t}^{\prime\prime} \in (x_{i}-\delta,x_{i}+\delta)$ is at least $(1/2) \delta (s-m) \geq \ell$.
By the union bound, this holds simultaneously for all $i \leq m$ with probability at least $1 - m \exp\{ - \delta (s-m) / 8 \} \geq 11/12$,
and on this event all of the $x_{ij}^{\prime}$ points exist.

The VC dimension of the collection of sets $\{ \{(x,h(x)) : x \in \X \} : h \in \H \}$ 
is $\dim$ (by definition of $\dim$).
Therefore, denoting
\begin{equation*}
A_{\ell,m} = 4 \frac{\dim \ln(2 e \ell / \dim) + \ln(96 m)}{\ell},
\end{equation*}
applying standard VC ``relative deviation'' bounds \cite{vapnik:74} (see Theorem 5.1 of \cite{boucheron:05}), 
to obtain a concentration inequality for the frequency of $(x_{ij}^{\prime},f(x_{ij}^{\prime})) \in \{ (x,h(x)) : x \in \X \}$,
holding for all $h \in \H$, we obtain that, 
for each $i \leq m$, with probability at least $1 - 1 / (12 m)$, 
if the points $x_{ij}^{\prime}$ exist, then 
every $h \in \H$ has
\begin{multline}
\label{eqn:vc-relative-ub}
\frac{1}{\ell} \sum_{j=1}^{\ell} \ind[ h(x_{ij}^{\prime}) \neq f(x_{ij}^{\prime}) ] 
\leq \P( h(x_{i1}^{\prime}) \neq f(x_{i1}^{\prime}) | x_{i} )
\\ + \sqrt{\P( h(x_{i1}^{\prime}) \neq f(x_{i1}^{\prime}) | x_{i} ) A_{\ell,m}} + A_{\ell,m}
\end{multline}
and
\begin{multline}
\label{eqn:vc-relative-lb}
\frac{1}{\ell} \sum_{j=1}^{\ell} \ind[ h(x_{ij}^{\prime}) \neq f(x_{ij}^{\prime}) ] 
\geq \P( h(x_{i1}^{\prime}) \neq f(x_{i1}^{\prime}) | x_{i} )
\\ - \sqrt{\P( h(x_{i1}^{\prime}) \neq f(x_{i1}^{\prime}) | x_{i} ) A_{\ell,m}}.
\end{multline}
The union bound implies this is true simultaneously for all $i \leq m$ with probability at least $11/12$.

Furthermore, for any $x_{i} \in (\delta,1-\delta)$, 
the conditional distribution of $x_{i1}^{\prime}$ given $x_{i}$ is ${\rm Uniform}(x_{i}-\delta,x_{i}+\delta)$, 
so that 
\begin{equation*}
\NSrm_{\delta}(f,x_{i};\H) = \inf_{h \in \H_{(x_{i},f(x_{i}))}} \P( h(x_{i1}^{\prime}) \neq f(x_{i1}^{\prime}) | x_{i} )
\end{equation*}
in the case $\H_{(x_{i},f(x_{i}))} \neq \emptyset$.
Thus, on the above events, for each $i \leq m$ with $x_{i} \in (\delta, 1-\delta)$ and $\H_{(x_{i},f(x_{i}))} \neq \emptyset$, 
taking the infimum over $h \in \H_{(x_{i},f(x_{i}))}$ on both sides of \eqref{eqn:vc-relative-ub} yields 
\begin{equation}
\label{eqn:NSxi-ub}
\widehat{\NSrm}_{\delta}(f,x_{i};\H) \leq \NSrm_{\delta}(f,x_{i};\H) + \sqrt{\NSrm_{\delta}(f,x_{i};\H) A_{\ell,m}} + A_{\ell,m}.
\end{equation}
For the other inequality, note that the left hand side of \eqref{eqn:vc-relative-lb} is nonnegative, so that the inequality remains valid if we include a maximum with $0$ on the right hand side.
Then noting that $x \mapsto x - \max\!\left\{\sqrt{x A_{\ell,m}}, 0\right\}$ is nondecreasing on $[0,1]$, 
we obtain, on the above events, for each $i \leq m$ with $x_{i} \in (\delta, 1-\delta)$ and $\H_{(x_{i},f(x_{i}))} \neq \emptyset$, 
\begin{equation}
\label{eqn:NSxi-lb}
\widehat{\NSrm}_{\delta}(f,x_{i};\H) \geq \NSrm_{\delta}(f,x_{i};\H) - \sqrt{\NSrm_{\delta}(f,x_{i};\H) A_{\ell,m}}.
\end{equation}
Both of these inequalities are trivially also satisfied in the case $\H_{(x_{i},f(x_{i}))} = \emptyset$.

Furthermore, since $\k \geq 80/\eps$, we have $2\delta \leq \frac{\eps^{3}}{16 \cdot 80}$, 
so that a Chernoff bound implies that, for an appropriately large choice of the numerical constant $c$, 
with probability at least $11/12$, 
\begin{equation*}
\frac{1}{m} \sum_{i=1}^{m} \ind[ x_{i} \notin (\delta,1-\delta) ] \leq \frac{\eps^{3}}{16 \cdot 65}.
\end{equation*}
Furthermore, note that since $\k \geq 80/\eps$, we have 
$\frac{\eps^{3}}{16 \cdot 65} = (\k-1) \frac{\delta}{2} \frac{\k}{\k-1} \frac{64 \eps}{16 \cdot 65} < (\k-1) \frac{\delta}{2} \frac{\eps}{16}$,
so that on the above event, 
\begin{equation}
\label{eqn:delta-ends-frac}
\frac{1}{m} \sum_{i=1}^{m} \ind[ x_{i} \notin (\delta,1-\delta) ] < (\k-1) \frac{\delta}{2} \frac{\eps}{16}.
\end{equation}

Additionally, since $\k \geq 80/\eps$, we have $(\k-1) \frac{\delta}{2} > \frac{\eps^{2}}{65}$, 
so that $m > \frac{c/65}{(\k-1) (\delta/2) \eps^{2}}$, which clearly also means $m > \frac{c/65}{(\k-1)(\delta/2)(1+\eps/4)\eps^{2}}$.
Therefore, another application of a Chernoff bound implies that, 
for an appropriately large choice of the numerical constant $c$, 
with probability at least $11/12$, 
\begin{equation}
\label{eqn:NSemp-ub}
\NS_{\delta}(f;\H) \leq (\k-1) \frac{\delta}{2} \implies \frac{1}{m} \sum_{i=1}^{m} \NSrm_{\delta}(f,x_{i};\H) \leq (\k-1) \frac{\delta}{2} \left( 1 + \frac{\eps}{33} \right)
\end{equation}
and
\begin{align}
& \NS_{\delta}(f;\H) > (\k-1) \frac{\delta}{2} \left( 1 + \frac{\eps}{4} \right) \notag
\\ & \!\implies \frac{1}{m} \sum_{i=1}^{m} \NSrm_{\delta}(f,x_{i};\H) 
> (\k\!-\!1) \frac{\delta}{2} \!\left( 1 + \frac{\eps}{4} \right) \!\left( 1 - \frac{\eps}{33} \right)
\geq (\k\!-\!1) \frac{\delta}{2} \!\left( 1 + \frac{7}{33}\eps \right)\!. \label{eqn:NSemp-lb}
\end{align}

The union bound implies that all four of the above events hold simultaneously with probability at least $2/3$.
Let us suppose all of these events indeed hold.
In this case, if $\NS_{\delta}(f;\H) \leq (\k-1) \frac{\delta}{2}$, then \eqref{eqn:NSxi-ub} and Jensen's inequality imply 
\begin{align*}
& \widehat{\NS}_{\delta}(f;\H) 
\\ & \leq \frac{1}{m} \sum_{i=1}^{m} \!\left( \NSrm_{\delta}(f,x_{i};\H) \!+\! \sqrt{\NSrm_{\delta}(f,x_{i};\H) A_{\ell,m}} \!+\! A_{\ell,m} \right) \!+\! \frac{1}{m} \sum_{i=1}^{m} \ind[ x_{i} \!\notin\! (\delta,1-\delta) ] 
\\ & \leq \left( \frac{1}{m} \sum_{i=1}^{m} \NSrm_{\delta}(f,x_{i};\H) \right) + \sqrt{\left( \frac{1}{m} \sum_{i=1}^{m} \NSrm_{\delta}(f,x_{i};\H) \right) A_{\ell,m}} 
\\ & \phantom{aaa} + A_{\ell,m} + \frac{1}{m} \sum_{i=1}^{m} \ind[ x_{i} \notin (\delta,1-\delta) ],
\end{align*}
and \eqref{eqn:delta-ends-frac} and \eqref{eqn:NSemp-ub} imply this is at most
\begin{equation*}
(\k-1) \frac{\delta}{2} \left( 1 + \frac{\eps}{33} \right) + \sqrt{(\k-1) \frac{\delta}{2} \left( 1 + \frac{\eps}{33} \right) A_{\ell,m}} + A_{\ell,m} + (\k-1) \frac{\delta}{2} \frac{\eps}{16}.
\end{equation*}
For appropriately large choices of the numerical constants $c',c''$, 
we can obtain $A_{\ell,m} \leq \frac{\eps^{4}}{65 \cdot 68 \cdot 33} \leq (\k-1) \frac{\delta}{2} \frac{\eps^{2}}{68 \cdot 33}$, 
so that the above expression is at most
\begin{equation*}
(\k-1) \frac{\delta}{2} \left( 1 + \frac{\eps}{33} + \sqrt{\left( 1 + \frac{\eps}{33} \right) \frac{\eps^{2}}{68 \cdot 33}} + \frac{\eps^{2}}{68 \cdot 33} + \frac{\eps}{16} \right)
\leq (\k-1) \frac{\delta}{2} \left( 1 + \frac{\eps}{8} \right),
\end{equation*}
which verifies the first claimed implication from the lemma.
On the other hand, for the second implication, if $\NS_{\delta}(f;\H) > (\k-1) \frac{\delta}{2} \left( 1 + \frac{\eps}{4} \right)$, 
then \eqref{eqn:NSxi-lb} and Jensen's inequality imply 
\begin{align*}
& \widehat{\NS}_{\delta}(f;\H) 
\\ & \geq \frac{1}{m} \sum_{i=1}^{m} \left( \NSrm_{\delta}(f,x_{i};\H) - \sqrt{\NSrm_{\delta}(f,x_{i};\H) A_{\ell,m}} \right) - \frac{1}{m} \sum_{i=1}^{m} \ind[ x_{i} \notin (\delta,1-\delta) ] 
\\ & \geq \left( \frac{1}{m} \sum_{i=1}^{m} \NSrm_{\delta}(f,x_{i};\H) \right) - \sqrt{\left( \frac{1}{m} \sum_{i=1}^{m} \NSrm_{\delta}(f,x_{i};\H) \right) A_{\ell,m}} 
\\ & \phantom{aaa}- \frac{1}{m} \sum_{i=1}^{m} \ind[ x_{i} \notin (\delta,1-\delta) ],
\end{align*}
and \eqref{eqn:delta-ends-frac} implies this is greater than 
\begin{equation*}
\left( \frac{1}{m} \sum_{i=1}^{m} \NSrm_{\delta}(f,x_{i};\H) \right) - \sqrt{\left( \frac{1}{m} \sum_{i=1}^{m} \NSrm_{\delta}(f,x_{i};\H) \right) A_{\ell,m}} - (\k-1) \frac{\delta}{2} \frac{\eps}{16}.
\end{equation*}
Now since our choices of constants $c,c',c''$ above imply 
$A_{\ell,m} \leq (\k-1) \frac{\delta}{2} \frac{\eps^{2}}{68 \cdot 33} \leq (\k-1) \frac{\delta}{2} \left( 1 + \frac{7}{33} \eps \right)$, 
and since $x \mapsto x - \sqrt{x A_{\ell,m}}$ is increasing for $x \geq A_{\ell,m}$, 
\eqref{eqn:NSemp-lb} implies the above expression is greater than 
\begin{align*}
& (\k-1) \frac{\delta}{2} \left( 1 + \frac{7}{33} \eps \right) - \sqrt{(\k-1) \frac{\delta}{2} \left( 1 + \frac{7}{33} \eps \right) A_{\ell,m}} - (\k-1) \frac{\delta}{2} \frac{\eps}{16}
\\ & \geq (\k-1) \frac{\delta}{2} \left( 1 + \frac{7}{33} \eps \right) - \sqrt{(\k-1) \frac{\delta}{2} \left( 1 + \frac{7}{33} \eps \right) (\k-1) \frac{\delta}{2} \frac{\eps^{2}}{68 \cdot 33}} - (\k-1) \frac{\delta}{2} \frac{\eps}{16}
\\ & = (\k-1) \frac{\delta}{2} \left( 1 + \frac{7}{33} \eps - \sqrt{\left( 1 + \frac{7}{33} \eps \right) \frac{\eps^{2}}{68 \cdot 33}} - \frac{\eps}{16} \right)
> (\k-1) \frac{\delta}{2} \left( 1 + \frac{\eps}{8} \right).
\end{align*}
This verifies the second claimed implication from the lemma, and thus completes the proof.
\end{proof}

We are now ready to finish describing the tester and prove its correctness.

\begin{theorem}
\label{thm:active-piecewise-function-tester}
If $\k \geq 80/\eps$,
then the procedure that outputs \accept~if $\widehat{\NS}_{\delta}(f;\H)$ $\leq (\k-1)\frac{\delta}{2}(1+\frac{\eps}{8})$,
and otherwise outputs \reject, 
is an $s$-sample $q$-query $\eps$-tester for the class $\F_{\k}(\H)$ of $\k$-piecewise $\H$ functions under the distribution ${\rm Uniform}(0,1)$,
for $s$ as defined above, and for 
$q = m (\ell+1)$ (where $m$ and $\ell$ are as defined above).
\end{theorem}
\begin{proof}
If $f \in \F_{\k}(\H)$, then 
Lemma~\ref{lem:NSInt-piecewise} implies it has
$\NS_{\delta}(f;\H) \leq (\k-1)\frac{\delta}{2}$,
so that Lemma~\ref{lem:NS-empirical} implies that with probability at least $2/3$, 
$\widehat{\NS}_{\delta}(f;\H) \leq$\break $(\k-1)\frac{\delta}{2} \left( 1 + \frac{\eps}{8} \right)$,
and hence the tester will output \accept.

On the other hand, if $f$ is $\eps$-far from every function in $\F_{\k}(\H)$, 
then Lemma~\ref{lem:NSnonInt-piecewise} implies that
$\NS_{\delta}(f;\H) > (\k-1)\frac{\delta}{2}(1+\frac{\eps}{4})$,
so that Lemma~\ref{lem:NS-empirical} implies that with probability at least $2/3$,
$\widehat{\NS}_{\delta}(f;\H) 
> (\k-1)\frac{\delta}{2}(1+\frac{\eps}{8})$,
and hence the tester will output \reject.

The claim about the number of samples and number of queries is immediate from the definition 
of the tester.
\end{proof}


Theorem~\ref{thm:active} immediately follows from this result for any $\k \geq 80/\eps$, since 
$m (\ell+1) = O\!\left(\frac{\dim}{\eps^{8}}\ln\!\left(\frac{1}{\eps}\right)\right)$
and 
$s = O\!\left(\frac{\dim \k}{\eps^{6}} \ln\!\left(\frac{1}{\eps}\right)\right)$.

For $\k < 80/\eps$, there is a trivial tester satisfying Theorem~\ref{thm:active}, 
based on the learn-then-validate technique of \cite{kearns:00}, which in fact works for any distribution $\Px$.
Specifically, in this case, we can take 
$s = \left\lceil \frac{c_{1} \dim \k}{\eps} \ln(2e\k) \ln\!\left(\frac{1}{\eps}\right) \right\rceil + \left\lceil \frac{c_{2}}{\eps} \right\rceil = O\!\left(\frac{\dim}{\eps^{2}} \ln^{2}\!\left(\frac{1}{\eps}\right)\right)$ 
iid $\Px$ samples, for appropriate numerical constants $c_{1},c_{2} \geq 1$, 
and query for the $f$ values for these $s$ samples (so $q = s$ here).
We then find a function $\hat{f} \in \F_{\k}(\H)$ consistent with $f$ on the first $\left\lceil \frac{c_{1} \dim \k}{\eps} \ln(2e\k) \ln\!\left(\frac{1}{\eps}\right) \right\rceil$ of these 
($\hat{f}$ chosen independently from the rest of the samples), if such a function $\hat{f}$ exists;
we then check whether $\hat{f}$ agrees with $f$ on at least 
$(1-\eps/2)\left\lceil \frac{c_{2}}{\eps} \right\rceil$ of the remaining $\left\lceil \frac{c_{2}}{\eps} \right\rceil$ samples.
If this $\hat{f}$ exists and satisfies this condition, then we output \accept, and otherwise we output \reject.
We can bound the graph dimension of $\F_{\k}(\H)$ as follows.
For any $n$ distinct points $x_{1},\ldots,x_{n} \in \reals$ and $n$ values $y_{1},\ldots,y_{n} \in \Y$, 
the number of distinct $\{(x,g(x)) : x \in \X\} \cap \{(x_{1},y_{1}),\ldots,(x_{n},y_{n})\}$ sets that can be realized by functions $g \in \F_{\k}(\H)$
is at most $\left(\frac{e n}{\dim}\right)^{\dim\k} \left(\frac{e n}{\k}\right)^{\k}$,
obtained by applying Sauer's lemma within each subset $(t_{j-1},t_{j}] \cap \{x_{1},\ldots,x_{n}\}$ and multiplying them 
to get at most $\left(\frac{e n}{\dim}\right)^{\dim\k}$ possible classifications for any fixed $t_{j}$ values, 
and then multiplying by a bound $\left(\frac{e n}{\k}\right)^{\k}$ on the number of ways to partition $\{x_{1},\ldots,x_{n}\}$ 
into at most $\k$ intervals.
Since $\left(\frac{e n}{\dim}\right)^{\dim\k} \left(\frac{e n}{\k}\right)^{\k}$ 
is strictly less than $2^{n}$ for any $n > 4 \dim \k \log_{2}( 2 e \k )$, 
the graph dimension of $\F_{\k}(\H)$ is at most $4 \dim \k \log_{2}( 2 e \k )$.
Thus, if $f \in \F_{\k}(\H)$, then standard VC bounds for the realizable case \cite{vapnik:74,blumer:89} imply that,
for an appropriate choice of the numerical constant $c_{1}$, with probability at least $5/6$, the function $\hat{f}$ 
will have $\Px(x : \hat{f}(x) \neq f(x)) < \eps/4$.  Also, for an appropriately large numerical constant $c_{2}$, 
a Chernoff bound implies that, with probability at least $5/6$, if $\Px(x : \hat{f}(x) \neq f(x)) < \eps/4$ then 
$\hat{f}$ will agree with $f$ on at least $(1-\eps/2)\left\lceil \frac{c_{2}}{\eps} \right\rceil$ of the last $\left\lceil \frac{c_{2}}{\eps} \right\rceil$ samples.
By the union bound, both of these events occur simultaneously with probability at least $2/3$, 
and the tester will output \accept~when they occur.
On the other hand, if $f$ is $\eps$-far from every function in $\F_{\k}(\H)$, 
then either $\hat{f}$ will not exist (in which case the tester outputs \reject), 
or else $\hat{f}$ is some function in $\F_{\k}(\H)$, so that $\Px(x : \hat{f}(x) \neq f(x)) > \eps$.
Therefore, for an appropriately large choice of the numerical constant $c_{2}$, 
a Chernoff bound 
implies that 
with probability at least $2/3$, if $\hat{f}$ exists,  
then it disagrees with $f$ on strictly more than $\frac{\eps}{2}\left\lceil \frac{c_{2}}{\eps} \right\rceil$ 
of the last $\left\lceil \frac{c_{2}}{\eps} \right\rceil$ samples, so that the tester will output \reject.

\section{Piecewise Constant Functions}
\label{sec:constant}

Next, we restrict focus to the special case of piecewise constant functions: 
that is, 
throughout this subsection, take $\H$ as the set of all \emph{constant} functions $\X \to \Y$.
In this case, we study both active and passive testing.  We construct a passive tester achieving 
the bound in Theorem~\ref{thm:constant}, as well as an active tester whose number of queries 
has an improved dependence on $\eps$ compared to Theorem~\ref{thm:active}.
Unlike the above general active tester, 
this construction follows more-closely the construction of testers for unions of intervals from \cite{yang:12},
and indeed recovers the same dependences on $\k$ and $\eps$ from that work, now for this more-general 
problem of testing piecewise-constant functions.

Since $\H$ is fixed as the set of constant functions $\X \to \Y$ in this section, 
we simply write $\NS_{\delta}(f)$ to abbreviate $\NS_{\delta}(f;\H)$, which (as we argue below) 
is consistent with the notion of noise sensitivity used in the prior literature \cite{yang:12}.
Fix any $\eps \in (0,1/2)$ and consider the case $\k \geq 80/\eps$.
Let 
$\delta = \frac{\eps^{2}}{32 \k}$, 
$m^{\prime} = \left\lceil \frac{c}{\eps^{4}} \right\rceil$,\break
$n = 1+\left\lceil 2 \sqrt{ \lceil 1/\delta \rceil} \right\rceil$,
and 
$s^{\prime} = 4 n m^{\prime}$, 
for an appropriate choice of numerical constant $c \geq 1$ from the analysis below.
Now the active and passive testers both sample $s^{\prime}$ points $z_{1}^{\prime},\ldots,z_{s^{\prime}}^{\prime}$ independent ${\rm Uniform}(0,1)$.
Let $t_{1},\ldots,t_{m^{\prime}}$ be the first $m^{\prime}$ distinct values $t$ in $\left\{ 1, \ldots s^{\prime} / n \right\}$
for which $\exists i,j \in  \{ (t-1) n + 1, \ldots, t n \}$ with $i < j$ 
and $|z_{i}^{\prime} - z_{j}^{\prime}| < \delta$ while $z_{i}^{\prime} \in (\delta,1-\delta)$, 
and for each $r \in \{1,\ldots,m^{\prime}\}$, denote by $i_{r}$ the smallest integer $i > (t_{r}-1)n$ with $\min_{j \in \{i+1,\ldots,t_{r}n\}} |z_{i}^{\prime}-z_{j}^{\prime}| < \delta$ while $z_{i}^{\prime} \in (\delta,1-\delta)$,
and denote by $j_{r}$ the smallest integer $j > i_{r}$ with $|z_{i_{r}}^{\prime} - z_{j}^{\prime}| < \delta$;
if there do not exist $m^{\prime}$ such values $t_{r}$, then the tester may output any response, and this is considered a failure event.
If these values do exist, then for each $r \leq m^{\prime}$, denote $z_{r} = z_{i_{r}}^{\prime}$ and $y_{r} = z_{j_{r}}^{\prime}$.
The active tester queries the $f$ values for $z_{r}$ and $y_{r}$, for each $r \leq m^{\prime}$,
whereas the passive tester (necessarily) queries the $f$ values for all $s^{\prime}$ points $z_{i}^{\prime}$.
Both testers then calculate the following quantity
\[
\widehat{\NS}_{\delta}^{\prime}(f) = \frac{1-2\delta}{m^{\prime}} \sum_{r=1}^{m^{\prime}} \ind[ f(z_{r}) \neq f(y_{r}) ]
\]
and outputs \accept~if $\widehat{\NS}_{\delta}^{\prime}(f) \leq (\k-1)\frac{\delta}{2}(1+\frac{\eps}{8})$,
and otherwise outputs \reject.

Comparing $\widehat{\NS}_{\delta}^{\prime}(f)$ to the quantity $\widehat{\NS}_{\delta}(f;\H)$ defined above, 
the main difference is that, for each of the points $z_{r}$, we use only a \emph{single} point $y_{r}$ 
sampled from $(z_{r}-\delta,z_{r}+\delta)$, rather than $\ell$ points.  For this reason, the total 
number of examples (both labeled and unlabeled) required to calculate $\widehat{\NS}_{\delta}^{\prime}(f)$ 
is significantly smaller than the number required to calculate $\widehat{\NS}_{\delta}(f;\H)$.
Nevertheless, in this special case of piecewise constant functions, we find that the 
guarantees we had for $\widehat{\NS}_{\delta}(f;\H)$ from Lemma~\ref{lem:NS-empirical} above 
remain valid for the quantity $\widehat{\NS}_{\delta}^{\prime}(f)$.  Specifically, we have the following lemma.

\begin{lemma}
\label{lem:NS-constant-empirical}
If $\k \geq 80/\eps$, 
then for an appropriate choice of the numerical constants $c$,
for any measurable function $f : \X \to \Y$, 
with probability at least $2/3$, the values $t_{1},\ldots,t_{m^{\prime}}$ exist, and 
\begin{align*}
\NS_{\delta}(f) \leq (\k-1) \frac{\delta}{2} & \implies \widehat{\NS}_{\delta}^{\prime}(f) \leq (\k-1) \frac{\delta}{2} \left( 1 + \frac{\eps}{8} \right)
\\ \NS_{\delta}(f) > (\k-1) \frac{\delta}{2} \left( 1 + \frac{\eps}{4} \right) & \implies \widehat{\NS}_{\delta}^{\prime}(f) > (\k-1) \frac{\delta}{2} \left( 1 + \frac{\eps}{8} \right).
\end{align*}
\end{lemma}
\begin{proof} 
The existence of the values $t_{1},\ldots,t_{m^{\prime}}$ (with high probability) follows from the so-called \emph{birthday problem} as follows.
Let $\beta = 1/\lceil 1/\delta \rceil$ and partition $[0,1)$ into disjoint intervals $[(i-1)\beta,i\beta)$, $i \in \{1,\ldots,1/\beta\}$.
For any $n$ iid ${\rm Uniform}(0,1)$ samples $w_{1},\ldots,w_{n}$ (for $n$ as defined above), 
the probability none of these intervals contains more than one $w_{j}$ value is 
$\prod_{j=1}^{n-1} (1-j\beta) \leq \exp\{ - \beta n (n-1)/2 \}$, 
and noting that $n \geq 1 + 2 \sqrt{ 1 / \beta }$, 
this is at most $e^{-2}$.
Furthermore, on the event that there exists at least one interval containing more than one $w_{j}$ value,
for $\hat{j}$ defined as the first $j$ such that $\exists j^{\prime} < j$ with $w_{j}$ and $w_{j^{\prime}}$ in the same interval, 
we note that the (unconditional) distribution of $w_{\hat{j}}$ is ${\rm Uniform}(0,1)$.  Therefore, with probability at least $1 - e^{-2} - 4\beta$, 
there exists some $i \in \{3,\ldots,(1/\beta)-2\}$ such that at least two $w_{j}$ values are in $[(i-1)\beta,i\beta)$.
Since $2\beta > 2\delta / (1+\delta) > \delta$, these intervals $[(i-1)\beta,i\beta)$ are strictly contained within $(\delta,1-\delta)$.
Furthermore, since $\beta < \delta < 1/32$, we have $1 - e^{-2} - 4\beta > 1 - e^{-2} - 1/8 > 1/2$.  
%
%
Thus, for each $t \in \{1,\ldots, s^{\prime} / n \}$, 
the sequence $\{z_{i}^{\prime} : i \in \{ (t-1) n + 1, \ldots, t n \} \}$
has probability at least $1/2$ of containing a pair $z_{i}^{\prime},z_{j}^{\prime}$ ($i < j$) with $|z_{i}^{\prime} - z_{j}^{\prime}| < \delta$ and $z_{i}^{\prime} \in (\delta,1-\delta)$.
In particular, the expected number of indices $t$ for which such a pair exists is at least $(1/2) s^{\prime} / n \geq 2m^{\prime}$.
Since these sequences are independent (over $t$),
a Chernoff bound implies that with probability at least $1 - \exp\{ - (1/2) (s^{\prime}/n) / 8 \} \geq 5/6$ (for any choice of $c \geq 4\ln(6)$), 
at least $m^{\prime}$ of these sequences contain such a pair, so that on this event the values $t_{1},\ldots,t_{m^{\prime}}$ indeed exist.

Next, note that since $\H$ is the set of \emph{constant} functions, 
for any $x,x^{\prime}$ and any $h \in \H_{(x,f(x))}$, we have $h(x^{\prime}) = h(x)$, 
which implies that for $x \sim {\rm Uniform}(0,1)$ and for $x^{\prime} \sim {\rm Uniform}(x-\delta,x+\delta)$ given $x$, 
we have $\NS_{\delta}(f) = \P(f(x) \neq f(x^{\prime}))$.
Furthermore, $z_{1}$ has distribution ${\rm Uniform}(\delta,1-\delta)$ and the conditional distribution of $y_{1}$ given $z_{1}$ is ${\rm Uniform}(z_{1}-\delta,z_{1}+\delta)$.
Therefore, 
\begin{align*}
\P( f(z_{1}) \neq f(y_{1}) ) 
& = \P( f(x) \neq f(x^{\prime}) | x \in (\delta,1-\delta) ) 
\\ & = \frac{1}{1-2\delta} \P( f(x) \neq f(x^{\prime}) \land x \in (\delta,1-\delta) ),
\end{align*}
and this rightmost quantity is at least as large as 
\begin{align*}
\frac{1}{1-2\delta} \left( \P( f(x) \neq f(x^{\prime}) ) - 2\delta \right) 
& = \frac{1}{1-2\delta} \left( \NS_{\delta}(f) - 2\delta \right)
\\ & \geq \frac{1}{1-2\delta} \left( \NS_{\delta}(f) - (\k-1) \frac{\delta}{2} \frac{\eps}{19} \right)
\end{align*}
and at most as large as 
\begin{equation*}
\frac{1}{1-2\delta} \P( f(x) \neq f(x^{\prime}) ) = \frac{1}{1-2\delta} \NS_{\delta}(f).
\end{equation*}
Thus, 
\begin{equation*}
\NS_{\delta}(f) \leq (\k-1) \frac{\delta}{2} \implies \P( f(z_{1}) \neq f(y_{1}) ) \leq \frac{1}{1-2\delta} (\k-1) \frac{\delta}{2}
\end{equation*}
and since $\frac{\eps}{4} - \frac{\eps}{19} = \frac{15}{76} \eps$, 
\begin{equation*}
\NS_{\delta}(f) > (\k-1) \frac{\delta}{2} \!\left( 1 + \frac{\eps}{4} \right) \implies \P( f(z_{1}) \neq f(y_{1}) ) > \frac{1}{1-2\delta} (\k-1) \frac{\delta}{2} \!\left( 1 + \frac{15}{76} \eps \right)\!.
\end{equation*}

Now recall that $\frac{1}{1-2\delta} (\k-1) \frac{\delta}{2} > \frac{\eps^{2}}{65}$,
and note that, if they exist, the pairs $(z_{r},y_{r})$ are iid over $r \leq m^{\prime}$.
Therefore (recalling the definition of $\widehat{\NS}_{\delta}^{\prime}(f)$ from above), 
a Chernoff bound implies that, for an appropriately large choice of the numerical constant $c$, 
with probability at least $5/6$, 
\begin{equation*}
\P( f(z_{1}) \neq f(y_{1}) ) \leq \frac{1}{1-2\delta} (\k-1)\frac{\delta}{2} \implies \widehat{\NS}_{\delta}^{\prime}(f) \leq (\k-1) \frac{\delta}{2} \left( 1 + \frac{\eps}{8} \right)
\end{equation*}
and
\begin{align*}
& \P( f(z_{1}) \neq f(y_{1}) ) > \frac{1}{1-2\delta} (\k-1)\frac{\delta}{2} \left( 1 + \frac{15}{76} \eps \right) 
\\ & \implies \widehat{\NS}_{\delta}^{\prime}(f) > (\k-1) \frac{\delta}{2} \left( 1 + \frac{15}{76} \eps \right) \left( 1 - \frac{\eps}{16} \right) 
> (\k-1) \frac{\delta}{2} \left( 1 + \frac{\eps}{8} \right). 
\end{align*}

Altogether, on the above two events, we have
\begin{equation*}
\NS_{\delta}(f) \leq (\k-1) \frac{\delta}{2} \implies \widehat{\NS}_{\delta}^{\prime}(f) \leq (\k-1) \frac{\delta}{2} \left( 1 + \frac{\eps}{8} \right)
\end{equation*}
and
\begin{equation*}
\NS_{\delta}(f) > (\k-1) \frac{\delta}{2} \left( 1 + \frac{\eps}{4} \right) \implies \widehat{\NS}_{\delta}^{\prime}(f) > (\k-1) \frac{\delta}{2} \left( 1 + \frac{\eps}{8} \right).
\end{equation*}

To complete the proof, we note that both of these events occur simultaneously with probability at least $2/3$ by the union bound.
\end{proof}

Finally, we have the following result on testing of piecewise constant functions.

\begin{theorem}
\label{thm:piecewise-constant-tester}
If $\k \geq 80/\eps$, then 
the procedure that outputs \accept~if $\widehat{\NS}_{\delta}^{\prime}(f) \leq (\k-1)\frac{\delta}{2}(1+\frac{\eps}{8})$,
and otherwise outputs \reject, is an $s^{\prime}$-sample $q$-query $\eps$-tester for the class of 
$\k$-piecewise constant functions under the distribution ${\rm Uniform}(0,1)$,
where $q = 2 m^{\prime}$ in the active testing variant and $q = s^{\prime}$ for the passive testing variant.
\end{theorem}
\begin{proof}
This proof is nearly identical to that of Theorem~\ref{thm:active-piecewise-function-tester}.
If $f \in \F_{\k}(\H)$, then 
Lemma~\ref{lem:NSInt-piecewise} implies it has
$\NS_{\delta}(f) \leq (\k-1)\frac{\delta}{2}$,
so that Lemma~\ref{lem:NS-constant-empirical} implies that with probability at least $2/3$, 
$\widehat{\NS}_{\delta}^{\prime}(f) \leq (\k-1)\frac{\delta}{2} \left( 1 + \frac{\eps}{8} \right)$,
and hence the tester will output \accept.

On the other hand, if $f$ is $\eps$-far from every function in $\F_{\k}(\H)$, 
then Lemma~\ref{lem:NSnonInt-piecewise} implies that
$\NS_{\delta}(f) > (\k-1)\frac{\delta}{2}(1+\frac{\eps}{4})$,
so that Lemma~\ref{lem:NS-constant-empirical} implies that with probability at least $2/3$,
$\widehat{\NS}_{\delta}^{\prime}(f) 
> (\k-1)\frac{\delta}{2}(1+\frac{\eps}{8})$,
and hence the tester will output \reject.

The number of samples and number of queries in the claim are immediate from the definition 
of the two testers.
\end{proof}

The upper bounds claimed in Theorem~\ref{thm:constant} immediately follow from this theorem, 
noting that 
$s^{\prime} = O\!\left(\frac{\sqrt{\k}}{\eps^{5}}\right)$, 
%
and also noting that for any $\k < 80/\eps$ we can obtain the result with 
the tester based on the learn-then-validate technique of \cite{kearns:00},
as described above at the end of Section~\ref{sec:active}.  In this latter case, the tester uses 
$\propto \frac{\k}{\eps} \ln(2e\k) \ln\!\left(\frac{1}{\eps}\right) = O\!\left(\frac{1}{\eps^{2}} \ln^{2}\!\left(\frac{1}{\eps}\right)\right)$ samples
(since $\dim = 1$ when $\H$ is the set of constant functions).

For the lower bound claimed in Theorem~\ref{thm:constant}, first note that our assumption of $\dim > 0$ 
implies, in the case of $\H$ the set of constant functions, that $|\Y| \geq 2$.
Therefore, we can \emph{reduce} testing unions of $\lfloor (\k-1)/2 \rfloor$ intervals to testing 
$\k$-piecewise constant functions by associating the binary labels $0$ and $1$ with any 
two distinct labels $y_{0},y_{1} \in \Y$.  Then any $\{0,1\}$-valued function $f_{01}$ can be mapped to a 
corresponding $\{y_{0},y_{1}\}$-valued function $f$, where 
$f$ is a $\k$-piecewise constant function if $f_{01}$ is a union of $\lfloor (\k-1)/2 \rfloor$ intervals, 
while $f$ is $\eps$-far from any $\k$-piecewise constant function if 
$f_{01}$ is $\eps$-far from any union of $\lfloor (\k-1)/2 \rfloor+1$ intervals.
This increase by one in the latter case is because the complement of a union of $\lfloor (\k-1)/2 \rfloor$ intervals 
is also $\k$-piecewise constant, but is possibly only representable as a union of $\lfloor (\k-1)/2 \rfloor+1$ intervals.  
To account for this increase by one, noting that the claim of an $\Omega\!\left(\sqrt{\k}\right)$ lower bound only regards large values of $\k$, 
if we suppose $\k > 1/\eps$, then any union of $\lfloor (\k-1)/2 \rfloor+1$ intervals is within distance $1/\k < \eps$ of a union of 
$\lfloor (\k-1)/2 \rfloor$ intervals.  
Thus, $f_{01}$ is $\eps$-far from any union of $\lfloor (\k-1)/2 \rfloor+1$ intervals
if it is $2\eps$-far from any union of $\lfloor (\k-1)/2 \rfloor$ intervals.
Altogether, we have that $f$ is $\eps$-far from any $\k$-piecewise constant function if 
$f_{01}$ is $2\eps$-far from any union of $\lfloor (\k-1)/2 \rfloor$ intervals.
So, by this reduction, the $\Omega\!\left(\sqrt{\lfloor (\k-1)/2 \rfloor}\right) = \Omega\!\left(\sqrt{\k}\right)$ lower bound 
of \cite{kearns:00,yang:12} for passive testing of unions of $\lfloor (\k-1)/2 \rfloor$ intervals implies a corresponding $\Omega\!\left(\sqrt{\k}\right)$ 
lower bound for testing $\k$-piecewise constant functions.\footnote{Technically, the result of \cite{kearns:00} establishes this lower bound 
for the problem of testing whether $f_{01}$ is a union of $n$ intervals or returns uniform random $\{0,1\}$ labels.  However, essentially the 
same argument would apply if, in the latter case, instead of random labels, we took $f_{01}$ to be a randomly-chosen binary \emph{function} 
based on a partition of $[0,1]$ into $n^{\prime} \gg n$ equal-sized regions, which would be at least $1/4 > 2\eps$ distance from any 
union of $n$ intervals with very high probability.}
This completes the proof of Theorem~\ref{thm:constant}.

\section{Open Problem on the Query Complexity of Testing Polynomials}
\label{sec:polynomials}

While the result above for active testing, when specialized to testing $\k$-piecewise degree-$\p$ polynomials, 
obtains the optimal dependence on the number of pieces $\k$, we were not able to 
show optimality in the degree $\p$.  This leads to an even more-basic question:\\

\noindent {\bf Open Problem:}
What is the optimal dependence on $\p$ in the 
query complexity of testing degree-$\p$ real-valued polynomials under $\Px = {\rm Uniform}(0,1)$?\\

This question is open at this time, for \emph{both} the active and passive property testing settings,
and indeed also for the stronger value-query setting 
(where the tester can query for $f(x)$ at any $x \in \X$).

There is a trivial $\p+1 + \frac{1}{\eps}\ln(3)$ upper bound for both active and passive testing, 
based on the learn-then-validate technique, 
since $\p+1$ random samples uniquely specify any degree-$\p$ polynomial (with probability one): 
that is, if we fit a degree-$\p$ polynomial $\hat{f}$ to the first $\p+1$ random points, 
then if $\target$ is a degree-$\p$ polynomial we would have $\hat{f} = f$ and thus the two would agree 
on the remaining $\frac{1}{\eps}\ln(3)$ points, in which case we may decide \accept; 
on the other hand, if $\target$ is $\eps$-far from all degree-$\p$ polynomials, then it is $\eps$-far from $\hat{f}$, and hence 
with probability at least $2/3$ at least one of $\frac{1}{\eps}\ln(3)$ random samples $x$ will have $\hat{f}(x) \neq \target(x)$, 
in which case we may decide \reject.

One might na\"{i}vely think that, since it is possible to fit any $\p+1$ values (at distinct $x$'s)
with a degree-$\p$ polynomial, a \emph{lower bound} of $\Omega(\p)$ should also hold.
However, we note that it is also possible to fit any $\k$ values (at distinct $x$'s) 
with a $\k$-piecewise constant function, and yet above we proved 
it is possible to test $\k$-piecewise constant functions using a number of queries 
with $\sqrt{\k}$ dependence on $\k$ by passive testing or \emph{independent} of $\k$ by active testing.
So the mere ability to fit $\p+1$ arbitrary values with a degree-$\p$ polynomial 
is not in-itself sufficient as a basis for proving a lower bound on the query complexity of 
testing degree-$\p$ polynomials.

\section*{Bibliography}
\bibliography{learning}

\end{document}